\documentclass[10pt,journal]{IEEEtran}

\usepackage{xspace,amsmath,epsfig,amssymb,subfigure,syntonly,times,amsthm}
\usepackage{psfrag,amsmath,color,bm,empheq,fancybox}
\usepackage{cite,footnote,xspace,syntonly,algorithm,algorithmic}
\usepackage{verbatim}
\input{mysymbol_old.sty}
\usepackage{needspace}

\newcommand{\myindentedparagraph}[1]{\needspace{1\baselineskip}\medskip \hangindent=11pt \hangafter=0 \noindent{\it #1.}}

\newtheorem{property}{\hspace{0pt}\bf Property}
\newtheorem{mycorollary}{\bf Corollary}

\def \reals    {{\mathbb R}}
\def \diag    {\text{\normalfont diag} }

\begin{document}


\title{The Dual Graph Shift Operator: Identifying the Support of the Frequency Domain}

\author{
	\IEEEauthorblockN{Geert Leus,~\IEEEmembership{Fellow,~IEEE}, Santiago Segarra,~\IEEEmembership{Member,~IEEE}, Alejandro Ribeiro,~\IEEEmembership{Member,~IEEE}, \\and Antonio G. Marques,~\IEEEmembership{Senior Member,~IEEE}}
\thanks{Work in this paper is supported by Spanish MINECO grants No TEC2013-
	41604-R, TEC2016-75361-R and USA NSF CCF-1217963. 
	G. Leus is with the Dept. of Electrical Eng., Math.
	and Comp. Science, Delft Univ. of Technology. S. Segarra is with the Inst. for Data, Systems and Society, Massachusetts Inst. of Technology.  A. Ribeiro are with the Dept. of Electrical and Systems Eng., Univ. of Pennsylvania. A. G. Marques is with the Dept. of Signal Theory and Comms., King Juan Carlos Univ. Emails: g.j.t.leus@tudelft.nl, segarra@mit.edu,  aribeiro@seas.upenn.edu, antonio.garcia.marques@urjc.es. 
	}}


\maketitle
\thispagestyle{empty}
\pagenumbering{arabic}

%

\begin{abstract}
Contemporary data is often supported by an irregular structure, which can be conveniently captured by a graph. Accounting for this graph support is crucial to analyze the data, leading to an area known as graph signal processing (GSP). The two most important tools in GSP are the graph shift operator (GSO), which is a sparse matrix accounting for the topology of the graph, and the graph Fourier transform (GFT), which maps graph signals into a frequency domain spanned by a number of graph-related Fourier-like basis vectors. This alternative representation of a graph signal is denominated the graph frequency signal. Several attempts have been undertaken in order to interpret the support of this graph frequency signal, but they all resulted in a one-dimensional interpretation. However, if the support of the original signal is captured by a graph, why would the graph frequency signal have a simple one-dimensional support? That is why, for the first time, we propose an irregular support for the graph frequency signal, which we coin the dual graph. The dual GSO leads to a better interpretation of the graph frequency signal and its domain, helps to understand how the different graph frequencies are related and clustered, {enables the development of better graph filters and filter banks, and facilitates the generalization of classical SP results to the graph domain.} 
\end{abstract}

\begin{IEEEkeywords}
	Graph signal processing, Dual graph shift operator, Frequency support, Graph Fourier transform, Duality.
\end{IEEEkeywords}




\section{Introduction}\label{sec:intro}

Graph signal processing (GSP) has emerged as an effective solution to handle data with an irregular support. Its approach is to represent this support by a graph, view the data as a signal defined on its nodes, and use algebraic and spectral properties of the graph to study the signals \cite{EmergingFieldGSP}. Such a data structure appears in many domains, including social networks, smart grids, sensor networks, and neuroscience. Instrumental to GSP are the notions of the graph shift operator (GSO), which is a matrix that accounts for the topology of the graph, and the graph Fourier transform (GFT), which allows the representation of graph signals in the so-called graph frequency domain. These tools are the fundamental building blocks for the development of compression schemes, filter banks, node-varying filters, windows, and other GSP techniques \cite{SamplingKovacevic_without_Moura_15,Oursampling_journal_2015,SandryMouraSPG_TSP13,isufi2017autoregressive,segarra2015graphfilteringTSP15,shuman2016vertex,marques2016stationaryTSP16}. 

Motivated by the practical importance of the GFT, some efforts have been made to establish a \textit{total ordering} of the graph frequencies \cite{EmergingFieldGSP,SandryMouraSPG_TSP14Freq,RabICASSP12_ApproxSignalsGraphs}, implicitly assuming a one-dimensional support for the graph frequency signal. Such an ordering translates into proximities between frequencies, which are critical for the definition of bandlimitedness and smoothness as well as for the design of sampling and (bank) filtering schemes. However, the basis vectors associated with frequencies that are close in such one-dimensional domains are often dissimilar and focus on completely different parts of the graph \cite{teke2016discreteuncertainty}, suggesting that a one-dimensional support is not descriptive enough to capture the similarity relationships between graph frequencies. {To overcome that limitation,} we propose the first description of the (not necessarily regular) support of a graph frequency signal by means of a graph (which we denominate as the \emph{dual graph}\footnote{This is not related to the graph theoretic notion of dual graph of a \textit{planar} graph $\ccalG$, which is a graph that has a vertex for each \textit{face} of $\ccalG$ \cite{yellen2003handbook}.}) and its corresponding dual GSO. This \textit{dual} GSO helps in describing the existing relations across frequencies, which can be ultimately leveraged to enhance existing vertex-frequency GSP schemes.

\section{The dual graph}
We start by reviewing fundamental concepts of GSP and then state formally the problem of identifying the dual GSO.

\subsection{Fundamentals of GSP}
Consider a graph $\ccalG$ of $N$ nodes or vertices with node set $\ccalN = \{ n_1, ..., n_N \}$ and edge set $\ccalE = \{ (n_{i}, n_{j}) \, | \, \text{$n_{i}$ is connected to $n_{j}$} \}$. The graph $\ccalG$ is further characterized by the so-called GSO, which is an $N\times N$ matrix $\bbS $ whose entries $[ \bbS ]_{ij}$ for $i \neq j$ are zero whenever nodes $n_i$ and $n_j$ are not connected. The diagonal entries of $\bbS$ can be selected freely and typical choices for the GSO include the Laplacian or adjacency matrices \cite{EmergingFieldGSP,SandryMouraSPG_TSP14Freq}. A graph signal defined on $\ccalG$ can be conveniently represented by a vector $\bbx = [ x_1,..., x_N ]^T \in \mathbb{C}^{N }$, where $x_i$ is the signal value associated with node $n_i$.

The GSO $\bbS$ -- encoding the structure of the graph -- is crucial to define the GFT and graph filters. The former transforms graph signals into a frequency domain, whereas the latter represents a class of local linear operators between graph signals. 
Assume for simplicity that the GSO $\bbS$ is normal, such that its eigenvalue decomposition (EVD) can always be written as $\bbS = \bbV \bbLambda \bbV^H$, where $\bbV$ is a unitary matrix that stacks the eigenvectors and $\bbLambda$ is a diagonal matrix that stacks the eigenvalues. 
To simplify exposition, we also assume the eigenvalues of the shift are simple (non-repeated), such that the associated eigenspaces are unidimensional. The eigenvectors $\bbV = [ \bbv_1, ..., \bbv_N ]$ correspond to the graph frequency basis vectors whereas the eigenvalues $\bblambda = \text{diag}(\bbLambda) = [ \lambda_1, ..., \lambda_N ]^T$ can be viewed as graph frequencies. With these conventions, the definitions of the GFT and graph filters are given next.

\begin{definition}\label{def_GFT} 
	Given the GSO $\bbS=\bbV\bbLambda\bbV^H$, the GFT of the graph signal $\bbx\in \mathbb{C}^{N }$ is $\tdbbx = [ \tdx_1, ..., \tdx_N ]^T := \bbV^H \bbx$.
	%
	%
\end{definition}
\begin{definition}\label{def_filters} 
	Given the GSO $\bbS=\bbV\bbLambda\bbV^H$, a graph filter $\bbH\!\in\! \mathbb{C}^{N \!\times\! N}\!$ of degree $L$ is a graph-signal operator of the form 
	\begin{eqnarray}\label{E:Filter_input_output_time}
	&\bbH=\bbH(\bbh,\bbS):=\sum_{l=0}^{L}h_l \bbS^l=\bbV\diag(\tdbbh)\bbV^H,&
	\end{eqnarray}
	where $\bbh:=[h_0,...,h_L]$ and $\tdbbh := \diag(\sum_{l=0}^{L}h_l \bbLambda^l)$.
\end{definition}
Definition \ref{def_GFT} implies that the inverse GFT (iGFT) is simply $\bbx = \bbV \tdbbx$. Vector $\bbh$ in Definition \ref{def_filters} collects the filter coefficients and $\tdbbh\in \mathbb{C}^{N } $ in \eqref{def_filters}  can be deemed as the frequency response of the filter.
{The particular case of the filter being $\bbH = \bbS$, so that $\tdbbh = \bblambda$, will be subject of further discussion in Section \ref{sec_finding_dual_eigenvecs}.}
Graph filters and the GFT have been shown useful for sampling, compression, filtering, windowing, and spectral estimation of graph signals \cite{SamplingKovacevic_without_Moura_15,Oursampling_journal_2015,SandryMouraSPG_TSP13,isufi2017autoregressive,segarra2015graphfilteringTSP15,shuman2016vertex,marques2016stationaryTSP16}.

\subsection{Support of the frequency domain}

The underlying assumption in GSP is that to analyze and process the graph signal $\bbx \in \mathbb{C}^{N }$ one has to take into account its graph support $\ccalG$ via the associated GSO $\bbS$. Moreover, according to Definition 1 the graph frequency signal $\tdbbx\in \mathbb{C}^{N }$ is an alternative representation of $\bbx$. Thus, a natural problem is the identification of the graph and the GSO corresponding to $\tdbbx$. More precisely, we are interested in finding the dual graph $\ccalG_f$ -- represented via the corresponding dual GSO $\bbS_f$ -- that characterizes the support of the frequency domain. 

Let $\ccalN_f= \{ n_{f,1},  ...,  n_{f,N}\}$ denote the node set of the dual graph $\ccalG_f$. Each element in  $\ccalN_f$ corresponds to a different frequency $(\lambda_i,\bbv_i)$, thus, the edge set $\ccalE_f$ indicates pairwise relations between the different frequencies. We interpret $\tdbbx$ as a signal defined on this dual graph, where $\tdx_i$ is associated with the node (frequency) $n_{f,i}$. 
As for the primal GSO, the EVD of the $N \!\times\! N$ matrix $\bbS_f$ {associated with $\ccalG_f$} will be instrumental to study $\tdbbx$. We start from the assumption that normality of $\bbS$ implies normality of $\bbS_f$. Later on, we will see that this assumption is valid. Due to normality, we have then that $\bbS_f = \bbV_f \bbLambda_f \bbV^H_f$, and thus the dual graph has (dual) frequency basis vectors $\bbV_f = [ \bbv_{f,1},..., \bbv_{f,N} ]$ and (dual) graph frequencies $\bblambda_f = \text{diag}(\bbLambda_f) = [ \lambda_{f,1},..., \lambda_{f,N} ]^T$ (cf. Fig. \ref{fig_primal_dual_shift_notation}). 

\vspace{.15cm}

\noindent \textbf{Problem statement.} \textit{Given the GSO $\bbS = \bbV \bbLambda \bbV^H$ find the dual GSO $\bbS_f = \bbV_f \bbLambda_f \bbV^H_f$.}

\vspace{.15cm}

To address this problem we postulate desirable properties that we want the dual GSO to satisfy. First, we start by identifying $\bbV_f$ (Section \ref{sec_finding_dual_eigenvecs}). We then proceed to determine $\bbLambda_f$ (Section \ref{sec_finding_dual_eigenvals}), which is a more challenging problem. 

\begin{figure}
	\centering
	\includegraphics[width=0.5\textwidth]{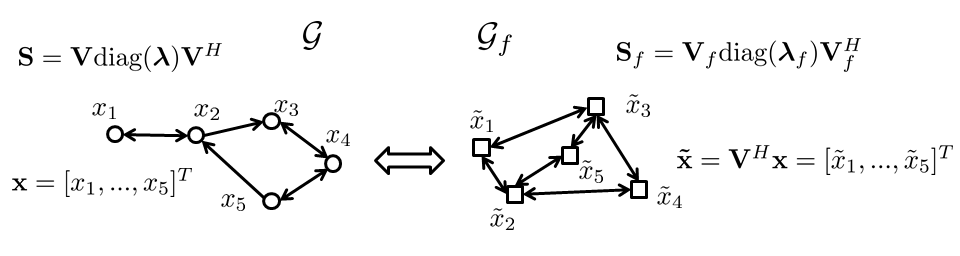}
	\vspace{-0.35in}
	\caption{ The primal graph (left) represents the support of the vertex domain, while the dual graph (right) represents the support of the frequency domain.}
	\vspace{-0.10in}
	\label{fig_primal_dual_shift_notation}
\end{figure}

\section{Eigenvectors of the dual graph}\label{sec_finding_dual_eigenvecs}

We want the GFT $\bbV_f^H$ associated with the dual graph to map $\tdbbx$ back to graph signal $\bbx$. Given that $\tdbbx = \bbV^H \bbx$ (cf. Definition~\ref{def_GFT}), the ensuing result follows.
\begin{property}\label{Prop:eigenvectors}
Given the primal GSO $\bbS=\bbV \bbLambda \bbV^H$, the eigenvectors of the dual GSO $\bbS_f$ are $\bbV_f = \bbV^H$.
\end{property}

As announced in the previous section, since we have that $\bbV_f^{-1}=\bbV_f^H$, then the dual shift $\bbS_f$ is normal too. 
With $\bbe_i\in\reals^N$ denoting the $i$th canonical basis vector (all entries are zero except for {the one corresponding to the $i$th \textit{node}}, which is one), then $\bbv_{f,i}$ can be written as $\bbv_{f,i}=\bbV^H\bbe_i=\tdbbe_i$, i.e., the GFT of the graph signal $\bbe_i$. Hence, {the dual frequency vector} $\bbv_{f,i}$ can be viewed as \textit{how} node $i$ expresses each of the primal graph frequencies, {revealing that each \textit{frequency of the dual} graph $\ccalG_f$ is related to a particular \textit{node of the primal}} graph $\ccalG$. 
{Moreover, we can also interpret the dual eigenvalues from a primal perspective. To that end, note that $\bblambda_f$ is the frequency response of the dual filter $\tdbbH=\bbS_f$ (cf. discussion after Definition~\ref{def_filters}); thus, the $i$th entry of $\bblambda_f$ can be understood as how strongly the primal value at the $i$th node $x_i$ is amplified when $\bbS_f$ is applied to $\tdbbx$.}

One interesting implication of Property 1 is that the dual of a Laplacian shift $\bbS = \bbV \bbLambda \bbV^H$ is, in general, not a Laplacian. Laplacian matrices require the existence of a constant eigenvector. Hence, for $\bbS_f$ to be a Laplacian, one of the rows of $\bbV$ -- corresponding to the columns of $\bbV_f$ -- needs to be constant, which in general is not the case. Another implication of Property 1 is the duality of the filtering and windowing operations, as shown next. 


\begin{mycorollary}
	Given the graph signal $\bbx\in \reals^N $ and the window $\bbw\in \reals^N $, define the windowed graph signal $\bbx_\bbw\in \reals^N $ as
	\begin{equation}\label{Eq:windowing_in_vertex}
	\bbx_\bbw = \diag(\bbw)\bbx.
	\end{equation}
	Then, recalling that $\tdbbx=\bbV^H\bbx$ and $\tdbbx_\bbw=\bbV^H\bbx_\bbw$, if $\bbS_f$ does not have repeated eigenvalues it holds that
	\begin{equation}\label{Eq:windowing_in_frequency}
	\tdbbx_\bbw\!=\!\bbH(\bbh_f,\bbS_f)\tdbbx,\;\;\;\text{with}\;\; \bbH(\bbh_f,\bbS_f)\!=\!\textstyle\sum_{l=0}^L h_{f,l}(\bbS_f)^l
	\end{equation} 
	for some $\bbh_f:=[h_{f,0},...,h_{f,L}]^T$ and $L\leq N-1$.
\end{mycorollary}
\begin{proof}
Substituting $\bbx_\bbw = \diag(\bbw)\bbx$ and $\bbx=\bbV \tdbbx$ into the definition of  $\tdbbx_\bbw$ yields $\tdbbx_\bbw=\bbV^H\diag(\bbw)\bbV\tdbbx$. This reveals that the mapping from $\tdbbx$ to 	$\tdbbx_\bbw$ is given by the matrix $\tdbbH=\bbV^H\diag(\bbw)\bbV$.  Since $\bbV^H$ is normal and unitary, $\bbV^H$ are the eigenvectors of $\tdbbH$ and $\bbw$ are its eigenvalues. Because $\bbV^H$ are also the eigenvectors of $\bbS_f$ (cf. Property 1), to show that $\tdbbH$ is a filter on $\bbS_f$ we only need to show that there exist coefficients $\bbh_f:=[h_{f,0},...,h_{f,N-1}]^T$ such that $\bbw= \diag(\sum_{l=0}^{N-1}h_{f,l} \bbLambda_f^l)$ [cf. \eqref{E:Filter_input_output_time}]. Defining $\bbPsi_f\in \mathbb{C}^{N \times N}$ as $[\bbPsi_f]_{i,l}=(\lambda_{f,i})^{l-1}$, the equality can be written as $\bbw=\bbPsi_f \bbh_f$. Since $\bbPsi_f$ is Vandermonde, if all the dual eigenvalues $\{\lambda_{f,i}\}_{i=1}^N$ are distinct, an $\bbh_f$ solving $\bbw=\bbPsi_f \bbh_f$ exists. \hfill $\blacksquare$
\end{proof}
The proof holds regardless of the particular $\bblambda_f$ and only requires $\bbS_f$ to have non-repeated eigenvalues. The corollary states that multiplication in the vertex domain is equivalent to filtering in the dual domain -- note that {the GSO of the filter in \eqref{Eq:windowing_in_frequency} is $\bbS_f$}. Clearly, when the entries of $\bbw$ are binary values, multiplying $\bbx$ by $\bbw$ acts as a windowing procedure preserving the values of $\bbx$ in the support $\bbw$, while discarding the information at the remaining nodes.

\section{Eigenvalues of the dual graph}\label{sec_finding_dual_eigenvals}

Given $\bbS=\bbV \diag(\bblambda) \bbV^H $ and using Property 1 to write the dual shift as $\bbS_f = \bbV^H \diag(\bblambda_f)\bbV$, the last step to identify $\bbS_f$ is to obtain $\bblambda_f$. Two different (complementary) approaches to accomplish this are discussed next.

\subsection{Axiomatic approach}\label{sec_axiomatic_approach}

{Our first approach is} to postulate properties that we want {the} dual shift $\bbS_f$ to satisfy, and then translate these properties into requirements on the dual eigenvalues $\bblambda_f$. We denominate these properties as \emph{axioms}, which we state next. {In the following,} $\bbP$ denotes an arbitrary permutation matrix.

\myindentedparagraph{(A1) Axiom of Duality} The dual of the dual graph is equal to the original graph 
\begin{equation}\label{E:axiom_duality}
(\bbS_f)_f = \bbS. \vspace{-0.10in}
\end{equation}
\vspace{-0.10in}
\myindentedparagraph{(A2) Axiom of Reordering} The dual graph is robust to reordering the nodes in the primal graph
\begin{equation}\label{E:axiom_reordering}
(\bbP \bbS \bbP^T)_f = \bbS_f.\vspace{-0.10in}
\end{equation}
\vspace{-0.10in}

\myindentedparagraph{(A3) Axiom of Permutation} Permutations in the EVD of the primal shift lead to permutations in the dual graph
\begin{equation}\label{E:axiom_permutation}
(\bbV \bbP \diag(\bbP^T \bblambda) \bbP^T \bbV^H)_f = \bbP^T (\bbV \diag(\bblambda) \bbV^H)_f \bbP.
\end{equation}

\vspace{0.1in}

\noindent Consistency with Property~\ref{Prop:eigenvectors} is encoded in the Axiom of Duality (A1). More precisely, since the GFT of the dual shift transforms a frequency signal $\tilde{\bbx}$ back into the graph domain $\bbx$, we want the associated shift to be recovered as well.The Axiom of Reordering (A2) ensures that the frequency structure encoded in the dual shift is invariant to relabelings of the nodes in the primal shift. Specifically, the frequency coefficients of a given signal $\bbx$ with respect to $\bbS$ should be the same as those of  $\bbx' = \bbP \bbx$ with respect to $\bbS' = \bbP \bbS \bbP^T$. Finally, since the nodes of the dual graph correspond to different frequencies, the Axiom of Permutation (A3) ensures that if we permute the eigenvectors (and corresponding eigenvalues) of $\bbS$, the nodes of the dual shift are permuted accordingly.

Axioms (A1)-(A3) impose conditions on the possible choices for the dual eigenvalues $\bblambda_f$. More precisely, let us define the function $h: \mathbb{C}^N \times \mathbb{C}^{N \times N} \to \mathbb{C}^N$, that computes the dual eigenvalues $\bblambda_f = h(\bblambda, \bbV)$ as a function of the eigen-decomposition of $\bbS$. In terms of $h$, axiom (A1) requires that
\begin{equation}\label{E:duality_h}
\bblambda =  h ( \bblambda_f, \bbV_f) = h ( h( \bblambda, \bbV), \bbV^H).
\end{equation}
In order to translate \eqref{E:axiom_reordering} into a condition on $h$, notice that $\bbP \bbS \bbP^T = \bbP \bbV \diag(\bblambda) \bbV^H \bbP^T$ so that $(\bbP \bbS \bbP^T)_f$ from Property~\ref{Prop:eigenvectors} must be equal to $\bbV^H \bbP^T \diag(\bblambda') \bbP \bbV$. Thus, for $(\bbP \bbS \bbP^T)_f$ to coincide with $\bbS_f$ we need $\bblambda' = \bbP \bblambda_f$ which ultimately requires that
\begin{equation}\label{E:reordering_h}
 h( \bblambda, \bbP \bbV) = \bblambda' = \bbP \bblambda_f = \bbP h( \bblambda, \bbV).
\end{equation}
Lastly, in order to find the requirement imposed by axiom (A3) on $h$, we again leverage Property~\ref{Prop:eigenvectors} to obtain $(\bbV \bbP \diag(\bbP^T \bblambda) \bbP^T \bbV^H)_f = \bbP^T \bbV^H \diag(\bblambda') \bbV \bbP$. It readily follows that to satisfy \eqref{E:axiom_permutation} we need $\bblambda' = \bblambda_f$, i.e.
\begin{equation}\label{E:permutations_h}
h{( \bbP^T \bblambda, \bbV \bbP)} = \bblambda' = \bblambda_f = h{( \bblambda, \bbV)}.
\end{equation}

It is possible to find a function $h$ that simultaneously satisfies \eqref{E:duality_h}-\eqref{E:permutations_h}, as shown next.

\begin{theorem}\label{th:dgraph}
The following class of functions satisfies~\eqref{E:duality_h}-\eqref{E:permutations_h}, leading to a generating method for dual graphs that abides by axioms (A1)-(A3)
\begin{equation}\label{eq:dgraph}
\bblambda_f = h( \bblambda, \bbV) = \bbD_f^{-1} \bbV \bbD \bblambda,
\end{equation}
where $\bbD = \text{diag} ( g(\bbv_1), \dots, g(\bbv_N) )$ and $\bbD_f = \text{diag} ( g(\bbv_{f,1}), \dots, g(\bbv_{f,N}) )$, with $g(\cdot)$ any permutation invariant function, i.e., $g(\bbP \bbx) = g(\bbx)$.
\end{theorem}

\begin{proof}
We show that~\eqref{eq:dgraph} satisfies~\eqref{E:duality_h},~\eqref{E:reordering_h}, and~\eqref{E:permutations_h}. Showing that~\eqref{E:duality_h} holds, requires only substituting~\eqref{eq:dgraph} into $h ( h( \bblambda, \bbV), \bbV^H)$, which yields
\[ h ( h( \bblambda, \bbV), \bbV^H) = \bbD^{-1} \bbV^H \bbD_f ( \bbD_f^{-1} \bbV \bbD \bblambda) = \bblambda. \]
In order to show \eqref{E:reordering_h}, notice that a permutation of the rows of $\bbV$ (the columns of $\bbV_f$) does not influence $\bbD$ and only permutes the diagonal entries of $\bbD_f$. Hence, we can write $h( \bblambda, \bbP \bbV)$ as
\begin{align} 
h( \bblambda, \!\bbP \bbV) \!=\! ( \bbP \bbD_f \bbP^T)^{-1} \bbP \bbV \bbD \bblambda \!=\! \bbP \bbD_f^{-1}  \bbV \bbD \bblambda \!=\! \bbP h( \bblambda, \!\bbV). \nonumber
\end{align}
Finally, since a permutation of the columns of $\bbV$ (the rows of $\bbV_f$) does not influence $\bbD_f$ and only permutes the diagonal entries of $\bbD$, we can write $h( \bbP^T \bblambda, \bbV \bbP)$ as [cf.~\eqref{E:permutations_h}]
\begin{align} 
h ( \bbP^T \bblambda, \bbV \bbP) &= \bbD_f^{-1} ( \bbV \bbP)  ( \bbP^T \bbD \bbP ) ( \bbP^T \bblambda)= \bbD_f^{-1} \bbV \bbD \bblambda \nonumber \\
& = h (\bblambda, \bbV). \hspace{4.8cm} \blacksquare\nonumber
\end{align}
\end{proof}

Note that Theorem \ref{th:dgraph} {proves} the existence of a class of eligible dual graphs, but it does not indicate that \emph{every} dual graph falls in this class. If we restrict ourselves to the class in \eqref{eq:dgraph}, which can be described by the function $g(\cdot)$, the simplest choice for $g(\cdot)$ is $g( \bbx )=1$. This results in $\bblambda_f =\bbV \bblambda$, but any power of any norm is also a valid function, i.e., $g(\bbx) = \| \bbx \|_p^q$. A possible policy to design a dual graph could be to select the function $g(\cdot)$ that optimizes a particular figure of merit (such as the minimization of the number of edges in the dual graph  $\ccalG_f$) yet keeping faithful to (A1)-(A3). This problem is discussed in more detail at the end of the following section. Furthermore, additional axioms can be imposed on $\bbS_f$ to further winnow the class of admissible functions $h$. A possible avenue, not investigated here, is to impose a desirable behavior of $\bbS_f$ with respect to the intrinsic phase ambiguity of the primal EVD.


\subsection{Optimization approach}\label{sec_optimization_approach}

A different (and complementary) approach is to find a dual shift $\bbS_f$ for which certain properties of practical relevance are promoted. For example, one may be interested in the sparsest $\bbS_f$. To be rigorous, consider that the primal shift $\bbS=\bbV\bbLambda\bbV^H$ is given. Then, upon setting $\bbv_{f,i} = \bbV^H\bbe_i$ (cf. Property 1), the dual shift $\bbS_f$ is found by solving
\begin{align}\label{E:general_problem_topology_id} 
\!\!\!\min_{ \{ \bbS_f, \, \bblambda_f \} } \ell (\bbS_f) \quad \text{s. to }\;\bbS_f\!=\!\textstyle\sum_{i =1}^N \lambda_{f,i}\bbv_{f,i}\bbv_{f,i}^H, \,\, \bbS_f \in \ccalS.
\end{align}
In the above problem the optimization variables are effectively the eigenvalues $\bblambda_f$, since the constraint $\bbS_f=\sum_{i =1}^N \lambda_{f,i}\bbv_{f,i}\bbv_{f,i}^H$ forces the columns of $\bbV_f$ to be the eigenvectors of $\bbS_f$. The objective $\ell$ promotes desirable network structural properties, such as sparsity or minimum-energy edge weights. The constraint set $\ccalS$ imposes requirements on the sought shift, such as each entry being non-negative or each node having at least one neighbor. This problem has been analyzed in detail in the context of network topology inference from nodal observations \cite{segarra2016topologyid}. 

One challenge of this approach is guaranteeing that the dual shift satisfies axioms (A1)-(A3), already deemed as desirable properties. For example, for axiom (A1) to hold, it is necessary for the original shift $\bbS$ itself to be optimal in the sense encoded by \eqref{E:general_problem_topology_id}. To elaborate on this, consider the normal unitary matrix $\bbU$ and the associated shift set $\ccalS_{\bbU} :=\{\bbS=\bbV \diag(\bbLambda) \bbV^H\;|\bbV=\bbU\;\text{and} \;\bblambda\in\mathbb{C}^N\}$. Moreover, let $\bbS^*$ denote the solution to \eqref{E:general_problem_topology_id}  when $\bbV_f=\bbU$ and $\bbS_f^*$ the solution when $\bbV_f=\bbU^H$. Then, it holds that: i) the dual shift for any $\bbS \in \ccalS_{\bbU}$ is given by $\bbS_f^*$, and ii) the dual of $\bbS_f^*$ is $\bbS^*$. Hence, $\bbS^*$ is the only element of $\ccalS_{\bbU}$ that guarantees that the dual of the dual is the original graph. Alternatively, one can see $\ccalS_{\bbU}$ as a shift class whose (canonical) representative is $\bbS^*$. With this interpretation any $\bbS \in \ccalS_{\bbU}$ is first mapped to $\bbS^*$ and then $\bbS^*$ serves as input for \eqref{E:general_problem_topology_id}. Under this assumption, the invertibility of the dual mapping is achieved.   

One important choice for the objective in \eqref{E:general_problem_topology_id} is to set $\ell (\cdot)=\|\cdot\|_0$, so that the goal is to find the sparsest shift (the one minimizing the number of pairwise relationships between the frequencies). Another interesting choice is to find a GSO that either minimizes the variability (maximizes the smoothness) of a given set of signals, or guarantees that the sum of the variability in the primal and dual domains does not exceed a given threshold -- which entails minor modifications to \eqref{E:general_problem_topology_id}.   

To ensure that the optimal $\bbS^*$ satisfies axioms (A1)-(A3), the approaches in Sections~\ref{sec_axiomatic_approach} and~\ref{sec_optimization_approach} can be combined. More specifically, one can solve \eqref{E:general_problem_topology_id} by optimizing over the class of admissible functions $g$; cf.~Theorem~\ref{th:dgraph}. This interesting (and more challenging) problem is left as future work.




\section{Illustrative Simulations}\label{sec:sims}
We provide a few simple examples illustrating that representations of the frequency domain that go beyond one dimension are of interest. 
To that end, we consider two primal graphs and compute their associated dual graphs by applying the methods in Sections \ref{sec_axiomatic_approach} (setting $\bblambda_f=\bbV\bblambda$) and \ref{sec_optimization_approach} (setting $\ell (\cdot)=\|\cdot\|_0$). The results are shown in Fig. \ref{fig_primal_dual_shifts}. The first row corresponds to the primal graph associated with the Discrete Cosine Transform (DCT) of type II \cite{moura2008DCT}, while the second row corresponds to an Erd\H{o}s-R\'enyi (ER) graph \cite{bollobas1998random} with $N=10$ and edge probability $p=0.15$. The first observation is that for none of the primal graphs the axiomatic approach gives rise to a sparse dual graph (cf. central column). This is important because the plotted dual graphs, which do not admit a one-dimensional representation, are legitimate representations of the pairwise interactions between the frequencies. 
The second observation is that the method promoting sparsity is able to find a \textit{very} sparse dual graph for the case of the DCT, but not for the ER graph. The sparse and regular dual graph obtained for the DCT serves as implicit validation of our approach and indicates that graphs with a very strong structure in the primal domain can be associated with strong and regular structures in the dual domain. In this extreme case, a one-dimensional representation could be argued to be sufficient. However, the dual shift corresponding to the ER graph demonstrates that the support of the frequency domain is, in general, more complicated and that structures that go beyond one-dimensional representations are required. These observations are confirmed for other types of graphs which, due to space limitations, are not presented here.

\begin{figure}
	\centering
	\includegraphics[width=0.5\textwidth]{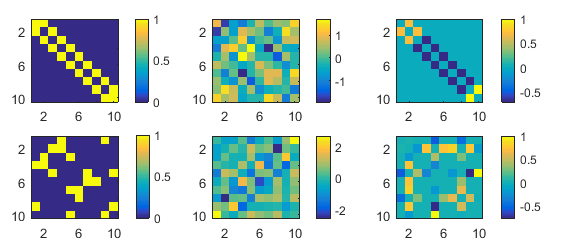}
	\caption{ The first \textit{row} corresponds to a DCT graph and the second an ER graph. The left \textit{column} plots the primal graph, the central the dual graph recovered using \eqref{eq:dgraph}, and the right column the dual graph recovered using \eqref{E:general_problem_topology_id}.}
	\vspace{-0.15in}
	\label{fig_primal_dual_shifts}
\end{figure}

\section{Conclusions and Open Questions}\label{sec:con}
This paper investigated the problem of identifying the support associated with the frequency representation of graph signals. Given the (primal) graph shift operator supporting graph signals of interest, the problem was formulated as that of finding a compatible \emph{dual} graph shift operator that serves as a domain for the frequency representation of these signals. We first identified the eigenvectors of the dual shift, showing that those correspond to how each of the nodes expresses the different graph frequencies.  We then proposed different alternatives to find the dual eigenvalues and characterized relevant properties that those eigenvalues must satisfy. Future work includes considering additional properties for the dual eigenvalues so that the size of feasible dual shift operators is reduced, and identifying additional results connecting the vertex domain with the frequency domain. The results in this paper constitute a first step towards understanding the structure of the signals in the frequency domain as well as developing enhanced GSP algorithms for signal compression, frequency grouping, filtering, and spectral estimation schemes.



\newpage
\begin{thebibliography}{10}

	
	\bibitem{EmergingFieldGSP}
	D.~Shuman, S.~Narang, P.~Frossard, A.~Ortega, and P.~Vandergheynst, ``The
	emerging field of signal processing on graphs: Extending high-dimensional
	data analysis to networks and other irregular domains,'' \emph{IEEE Signal
		Process. Mag.}, vol.~30, no.~3, pp. 83--98, May 2013.
	
	\bibitem{SamplingKovacevic_without_Moura_15}
	S.~Chen, R.~Varma, A.~Sandryhaila, and J.~Kova{\v{c}}evi{\'c}, ``Discrete
	signal processing on graphs: Sampling theory,'' \emph{IEEE Trans. Signal
		Process.}, vol.~63, no.~24, pp. 6510 -- 6523, Dec. 2015.
	
	\bibitem{Oursampling_journal_2015}
	A.~G. Marques, S.~Segarra, G.~Leus, and A.~Ribeiro, ``Sampling of graph signals
	with successive local aggregations,'' \emph{IEEE Trans. Signal Process.},
	vol.~64, no.~7, pp. 1832 -- 1843, Apr. 2016.
	
	\bibitem{SandryMouraSPG_TSP13}
	A.~Sandryhaila and J.~Moura, ``Discrete signal processing on graphs,''
	\emph{IEEE Trans. Signal Process.}, vol.~61, no.~7, pp. 1644--1656, Apr.
	2013.
	
	\bibitem{isufi2017autoregressive}
	E.~Isufi, A.~Loukas, A.~Simonetto, and G.~Leus, ``Autoregressive moving average
	graph filtering,'' \emph{IEEE Trans. Signal Process.}, vol.~65, no.~2, pp.
	274--288, 2017.
	
	\bibitem{segarra2015graphfilteringTSP15}
	S.~Segarra, A.~G. Marques, and A.~Ribeiro, ``Distributed linear network
	operators using graph filters,'' \emph{arXiv preprint arXiv:1510.03947},
	2015.
	
	\bibitem{shuman2016vertex}
	D.~I. Shuman, B.~Ricaud, and P.~Vandergheynst, ``Vertex-frequency analysis on
	graphs,'' \emph{Applied and Computational Harmonic Analysis}, vol.~40, no.~2,
	pp. 260--291, 2016.
	
	\bibitem{marques2016stationaryTSP16}
	A.~G. Marques, S.~Segarra, G.~Leus, and A.~Ribeiro, ``Stationary graph
	processes and spectral estimation,'' \emph{arXiv preprint arXiv:1603.04667},
	2016.
	
	\bibitem{SandryMouraSPG_TSP14Freq}
	A.~Sandryhaila and J.~Moura, ``Discrete signal processing on graphs: Frequency
	analysis,'' \emph{IEEE Trans. Signal Process.}, vol.~62, no.~12, pp.
	3042--3054, June 2014.
	
	\bibitem{RabICASSP12_ApproxSignalsGraphs}
	X.~Zhu and M.~Rabbat, ``Approximating signals supported on graphs,'' in
	\emph{IEEE Intl. Conf. Acoust., Speech and Signal Process. (ICASSP)}, Mar.
	2012, pp. 3921--3924.
	
	\bibitem{teke2016discreteuncertainty}
	O.~Teke and P.~P. Vaidyanathan, ``Discrete uncertainty principles on graphs,''
	in \emph{Asilomar Conf. on Signals, Systems, and Computers}, Pacific Grove,
	CA, Nov. 7-10 2016, pp. 1475--1479.
	
	\bibitem{yellen2003handbook}
	J.~Yellen and J.~L. Gross, \emph{Handbook of Graph Theory (Discrete Mathematics
		and Its Applications)}.\hskip 1em plus 0.5em minus 0.4em\relax CRC press,
	2003.
	
	\bibitem{segarra2016topologyid}
	S.~Segarra, A.~G. Marques, G.~Mateos, and A.~Ribeiro, ``Network topology
	inference from spectral templates,'' \emph{arXiv preprint
		arXiv:1608.03008v1}, 2016.
	
	\bibitem{moura2008DCT}
	M.~Puschel and J.~M.~F. Moura, ``Algebraic signal processing theory: 1-d
	space,'' \emph{IEEE Trans. Signal Process.}, vol.~56, no.~8, pp. 3586--3599,
	Aug 2008.
	
	\bibitem{bollobas1998random}
	B.~Bollob{\'a}s, \emph{Random Graphs}.\hskip 1em plus 0.5em minus 0.4em\relax
	Springer, 1998.
	
\end{thebibliography}
\end{document}